\documentclass{article}
\usepackage[english]{babel}
\usepackage[latin1]{inputenc} 
\usepackage{amsmath,amssymb}
\usepackage{graphicx}
\usepackage{psfrag}
\usepackage{amsfonts}          
\usepackage{color}
\usepackage{helvet}
\usepackage{mathrsfs}
\usepackage{authblk}

\def\order{\gamma}
\newcommand{\R}{\mathbb{R}}
\newcommand{\Cc}{\mathbb{C}}
\newcommand{\N}{\mathbb{N}}
\def\n{\nu}
\newcommand{\aaa}{\mathsf{a}}
\newcommand{\C}{\mathscr{C}}
\newcommand{\Dg}{{\textsc{\textbf{D}}}}   
\newcommand{\Ig}{{\textsc{\textbf{I}}}}   
\renewcommand{\(}{\left( }
\renewcommand{\)}{\right) }
\newcommand{\resp}{{\emph{resp.{ }}}}
\newcommand{\diag}[1] {\mathsf{diag}\{ #1 \}}
\newcommand{\x}{\times}

\newcommand{\tab}{\quad}

\newcommand{\defref}[1]{Definition \ref{#1}}

\newcommand{\figref}[1]{Figure \ref{#1}}

\newtheorem{thm}{Theorem}[section] 
\newtheorem{prop}{Proposition}[section] 
\newtheorem{lem}{Lemma}[section] 
\newtheorem{defn}{Definition}[section] 
\newtheorem{rem}{Remark}[section] 
\newtheorem{alg}{Algorithm}[section]

\begin{document}

\title{Flatness for linear fractional systems with application to a thermal system}

\author[*]{St\'{e}phane Victor} 
\author[*]{Pierre Melchior}  
\author[**]{Jean L\'{e}vine} 
\author[*]{Alain Oustaloup}
    
\affil[*]{IMS -- UMR 5218 CNRS, 
IPB/Enseirb-Matmeca \\ Universit\'{e} de Bordeaux,
351 cours de la Lib\'{e}ration, 33405 Talence cedex -- France  --
    \{stephane.victor, pierre.melchior, alain.oustaloup\}@ims-bordeaux.fr}    
    
\affil[**]{CAS-Mathématiques et Systèmes, MINES-ParisTech,
35 rue Saint-Honoré, 77300 Fontainebleau -- France --
    jean.levine@mines-paristech.fr}
        
\date{}

\maketitle

\begin{paragraph}{Keywords} Polynomial matrices, fractional systems, differential flatness, thermal system, trajectory planning.
\end{paragraph}

\begin{abstract}
This paper is devoted to the study of the flatness property of linear time-invariant fractional systems. In the framework of polynomial matrices of the fractional derivative operator,
we give a characterization of fractionally flat outputs and a simple algorithm to compute them. We also obtain a characterization of the so-called fractionnally $0$-flat outputs. 
We then present an application to a two dimensional heated metallic sheet, 
whose dynamics may be  approximated by a fractional model of order 1/2. The trajectory planning of the temperature at a given point of the metallic sheet is obtained thanks to the fractional flatness property, without integrating the system equations. The pertinence of this approach is discussed on simulations.
\end{abstract}


\section{Introduction}\label{sec:intro}

\emph{Fractional} (or \emph{non integer}) differential equations prove to be particularly pertinent for the modeling of some classes of systems such as thermal \cite{Old74,Bat00}, electro-chemical \cite{Dar97}, viscoelastic \cite{Mor02}, nuclear magnetic resonance \cite{Mag08} or bio\-logical ones \cite{Som07a}, whereas integer models may lead to repre\-sentations with an excessive number of state variables or poorly reproducing some dynamical aspects \cite{Mal09}. 
Fractional calculus dates back to Euler, Leibniz, Liouville, Riemann and many o\-thers (see e.g. \cite{Pod99a,Sam93}).
For relations between Mikusi\'{n}ski's operational calculus and fractional integration the reader may refer to \cite{Bat75}. In control applications, se\-minal works on modeling, identification and robust control design may be found in Oustaloup \cite{Ous95,Ous14}. For
controllability, observability, mini\-mal realization aspects of fractional systems in the algebraic framework of mo\-dule theory via Mikusi\'{n}ski's calculus, one may refer to Fliess and Hotzel \cite{Hot98b,Fli97a}. {Stability and stabilization in relation to controllability and observability have  been extensively studied by Matignon and d'Andrea-Novel \cite{Mat96a,Mat97}. }

Motion planning for fractional systems becomes a major issue if working around an equilibrium point appears to be too restrictive. If we need to properly design a feasible trajectory to be tracked, no systematic approach is presently available in this context. 
One possible approach consists in extending the so-called flatness-based trajectory design (see \cite{Fli92,Fli95,Fli99}, the books \cite{Rud03,Rot04,Sir04,Lev09} and the references therein) to fractional systems. 
Recall that, roughly speaking, a system described by \emph{ordinary differential equations} is said to be \emph{differentially flat} if  and only if there exists an output vector, called \emph{flat output}, of the same dimension as the control vector, such that all system variables can be expressed as functions of this output and its successive time derivatives in finite number. 
{The role played by flat outputs for motion planning is thus clear: all system trajectories being parametrized by such an output, which, in addition, does not need to satisfy any dif\-ferential equation, it suffices to construct an output curve by interpolation and deduce the desired state trajectory which, moreover, may be obtained without integrating the system equations, thus making this solution particularly simple and elegant.}
{Our aim in this paper is therefore (1) to extend this property to fractional systems, a notion that will be called \emph{fractional flatness}, (2) to characterize \emph{fractionally flat systems} and \emph{fractionally flat outputs}, (3) to provide an algorithm to compute such flat outputs, and (4) to show the usefulness of our approach on the motion planning example of a thermal system.}
Preliminary results on fractional flatness have been presented by some of the authors \cite{Vic10a,Vic11b} without algebraic foundations and in the particular case where the matrix $B$ of System (\ref{eq:poly_SSR}) (see Section~\ref{subsec:fracsys}) is a 0-degree polynomial matrix, leading to a less precise characterization of the so-called \emph{defining matrices}.  

{ The theoretical contributions of this paper are: a rigo\-rous algebraic definition of the fractional flatness pro\-perty for linear time invariant fractional systems, and its characterization in the framework of polynomial matrices of the fractional derivative operator; as a by-product, we recover the equivalence between  fractional flatness and controllability (see e.g. \cite{Hot98b,Fli97a});  then a simple algorithm to compute fractionnally flat outputs is obtained as well as a specia\-lization of the latter results to the notion of fractional $0$-flatness.

Concerning the applications, we extend previous results of Rudolph \cite{Rud03}, Laroche et al. \cite{Lar98,Lar00}
and others, establishing a link between the heat equation in one dimension, with various boundary conditions, and an approximate fractional system representation of order $\frac{1}{2}$, to the heating of a two dimensional sheet, modeled by the heat equation in the positive orthant of the $(x,y)$-plane and controlled by the heat density flux through the $y$-axis; we further show that the associated fractional approximation is flat with an easily obtained flat output that, to the authors know\-ledge, has no simple physical interpretation as a variable related to the aforementioned heat equation; finally, we show how to use these results to plan rest-to-rest trajectories of the temperature at a given point of the metallic sheet.
}

After recalling the basics of fractional calculus in Section \ref{sec:intro}, Section \ref{sec:FC} presents the fractional polynomial algebraic framework. Then, the notion of fractional flatness is introduced in Section \ref{sec:Sortie_ plate_Matrice_poly}. In Section \ref{sec:thermal_app}, a two-dimensional thermal application is presented in simulation, 
to finally conclude in section \ref{sec:conclu}.

\section{Recalls on Fractional Calculus}\label{sec:FC}
\subsection{Fractional derivative}

Let $\order \in \R_{+}$,
$n=\lceil \order \rceil = E(\order) +1 = \min \{ k \in \N \vert k > \order\}$, where $\lceil. \rceil$ (resp. $E$) is the ceiling operator (resp. the integer part), and $\n \in [0,1[$ given by
$\n = n-\order$.
Let $\aaa$ be a given arbitrary real number and $f \in  \C^{\infty}([\aaa,+\infty[)$, the set of infinitely continuously differentiable functions from $[\aaa,+\infty[$ to $\R$. We denote by $f^{(k)}(t)$ the ordinary $k$-th order derivative of $f$ with respect to $t$ for every $k\in \N$.

The \emph{Riemann-Liouville derivative}\footnote{This definition is generally credited to Riemann when $\aaa\neq 0$ , to Liouville when $\aaa=-\infty$, and to Riemann-Liouville when $\aaa=0$. Here since $\aaa$ is arbitrary, we gather the Riemann and Riemann-Liouville cases in one and give the name ``Riemann-Liouville'' to this generic concept, the role played by the bound $\aaa$ being of minor importance.}, or more simply \emph{fractional derivative}, 
of order $\order=n-\n$, denoted by $\Dg^{\order}_\aaa$, is defined as  the
$n$-th order (ordinary) derivative of the Cauchy integral of order $\n$ of $f$ at time $t$ \cite{Mil93}:
\begin{equation}\label{eq:def_frac_int}
	\begin{aligned}
     \Dg^{\order}_\aaa f(t)  &= \Dg^{n}
        \left(\Ig^{\n }_\aaa f(t)
    \right)\\
    & \triangleq
    \left( {\frac{d}{{dt}}} \right)^{n }
    \left( \frac{1}{{\Gamma
    \left( {\n} \right)}}
    {\int_\aaa^t {\frac{f\left( \tau  \right)d\tau }{\left( {t
    - \tau } \right)^{1-\n} }} } \right)
    \end{aligned}
\end{equation}
where $\Ig^{\n }_\aaa f(t)$ is called the \emph{fractional primitive} of $f$ and the Euler's $\Gamma$
function, defined by:
\begin{equation}\label{eq:Gamma}
    \Gamma(x) = \int_0^\infty  {e^{ - t} t^{x - 1} dt},
    \quad\forall x\in\R^* \setminus\N^-,
\end {equation}
is the generalized factorial ($\forall n \in \N, \Gamma(n+1)= n!$). Note that $\n\Gamma(\n) = \Gamma(\n + 1)$ for all $\nu \in \R_{+}$.
If $\order = n \in \N$, the fractional  and ordinary derivatives coincide ($\Dg^\order_\aaa f(t)= \Dg^n f(t)$) and if $\order<0$, $\Dg^\order_\aaa f(t) = \Ig^{-\order}_\aaa f(t)$.
Indeed, when $n=0$, the differentiation (\ref{eq:def_frac_int}) boils down to an integration.

Let us denote  
\begin{equation}\label{eq:Ha}
 \mathfrak{H}_\aaa \mathop  =  \limits^\Delta \left\{ f:\R\mapsto\R| f \in \C^{\infty}([\aaa,+\infty[), f(t) = 0, \forall t\leq \aaa\right\}.
 \end{equation}
As a consequence of Propositions~\ref{prop:DnInu_InuDn} and \ref{prop:Dcom} of the Appendix (see also e.g. \cite{Pod99a}), $\Dg_\aaa^\order$ is an endomorphism from $\mathfrak{H}_\aaa$
to itself and $\mathfrak{H}_\aaa$ 
may be considered as the domain of $\Dg_\aaa^\order$. For simplicity's sake, the notation $\Dg_\aaa^\order$ is used in place of ${\Dg_\aaa^\order}_{\big | \mathfrak{H}_\aaa}$.
We consider arbitrary polynomials of the indeterminate $\Dg_\aaa^\order$ with real coefficients, of the form $\sum\limits_{k = 0}^K{{c_k} \Dg_\aaa^{{k \order} } }$ and call them $\Dg_\aaa^\order$-polynomials.
Their degree is defined as usually.

We denote by $\R\left[\Dg_\aaa^\order\right]$ the set of such $\Dg_\aaa^\order$-polynomials (again whose domain is restricted to $\mathfrak{H}_\aaa$) 
endowed with the usual addition and multiplication of polynomials (denoted as usual by $+$ and $\times$). The reader may immediately verify that $\(\R\left[\Dg_\aaa^\order\right],+,\times\)$ is a (commutative) principal ideal domain. 
The main properties of $\R\left[\Dg_\aaa^\order\right]$ are recalled in the Appendix.

\subsection{$\Dg_\aaa^\order$-polynomial matrices}

Interpreting the fractional derivative operator in terms of Mikusi\'{n}ski's operational calculus as in Battig and Kalla \cite{Bat75}, a system theoretic approach has been developed by Fliess and Hotzel \cite{Fli97a} where the field of Mikusi\'{n}ski's operators $\mathcal{M}$ is defined as the field of fractions of the commutative integral domain $\mathcal{C}$ of continuous functions defined over $\left[ 0, \infty\right[$ endowed with the addition and convolution product \cite{Mik83,Fli97a}. 
$\mathcal{M}$ can also be considered as an $\R\left[s^\order\right]$-module, where $s^\order$ is the Laplace operator associated to $\Dg_\aaa^\order$.

Note that for all $A\in\mathcal{M}$ and all $f \in \mathcal{C}$, the result $Af$ is defined in the sense of distributions, whereas in  the previous approach, namely in $\R\left[\Dg_\aaa^\order\right]$, the action  of a polynomial of $ \Dg_\aaa^\order$ applied to a function of $ \mathfrak{H}_\aaa$ is always well-defined in $ \mathfrak{H}_\aaa$. 
However the results stated in the remainder of this paper are equally well established in both approaches. We only state them in the setting of $\R\left[\Dg_\aaa^\order\right]$. Their (straightforward) adaptation to $\mathcal{M}$ is left to the reader. 
{Note also that the parameterization algorithms of \cite{Chy05} might be adaptable to this fractional context to yield comparable results, though apparently in a more indirect way.}

If $p,q \in \N$, we call $\R\left[\Dg_\aaa^\order\right]^{p\times q}$ the set of $\Dg_\aaa^\order$-polynomial matrices of size $(p\times q)$, i.e. whose entries are $\Dg_\aaa^\order$-polynomials.
When $p=q$, the group of unimodular $\Dg_\aaa^\order$-polynomial matrices, $GL_{p}\(\R\left[\Dg_\aaa^\order\right]\)$, defines the set of invertible (square) $\Dg_\aaa^\order$-polynomial matrices whose inverse is also a $\Dg_\aaa^\order$-polynomial matrix.
We denote by $I_p$ the $p\times p$ identity matrix and by $0_{p\times q}$ the $p\times q$ zero matrix.
$\Dg_\aaa^\order$-polynomial matrices enjoy the following important property \cite[Chap. 8]{Coh85}:
\begin{thm}[Smith diagonal decomposition]\label{th:smith}
                   Given a matrix \sloppy$A\in \R\left[\Dg_\aaa^\order\right]^{p\times q}$, with $p \leq q$ (\resp{} $p\geq q$), there exist two matrices $S\in GL_{p}\(\R\left[\Dg_\aaa^\order\right]\)$ and $T\in GL_{q}\(\R\left[\Dg_\aaa^\order\right]\)$ such that:
                    \begin{equation}\label{eq:smith}
                    SAT = \left[ {\Delta \tab 0_{p,q-p}} \right]\tab (\resp{} = \left[ {\begin{array}{*{20}c}
                    \Delta   \\
                    0_{p-q,q}  \\
                    \end{array}} \right]),
                    \end{equation}
                    where $\Delta = \diag{\delta _1 , \ldots ,\delta _\sigma  ,0, \ldots ,0} \in \R\left[\Dg_\aaa^\order\right]^{p\times p}$ (resp. $\R\left[\Dg_\aaa^\order\right]^{q\times q}$). In $\Delta$, the integer $\sigma \leq \min (p,q)$ is the \emph{rank} of $A$ and every non zero $\Dg_\aaa^\order$-polynomial $\delta _i$, for $i=1,\dots,\sigma$, is a divisor of $\delta _j$ for all $\sigma \ge j \ge i$.
                \end{thm}

 This decomposition over $\R\left[\Dg_\aaa^\order\right]$ is done in practice by the same algorithm as the one over $\R\left[s\right]$, $s$ being a complex variable. It consists in computing a diagonal form by repeatedly using greatest common  divisors (see e.g. \cite{Gan60,Wol74}. {See also \cite{Chy05} for  more general algorithms in the context of Ore algebras).}

\begin{defn}[Hyper-regularity \cite{Lev09}] \label{def:hyper-regularity}
Given a matrix \sloppy$A\in  \R\left[\Dg_\aaa^\order\right]^{p\times q}$, we say that $A$ is hyper-regular if, and only if, in (\ref{eq:smith}), we have $\Delta = I_{\min (p,q)}$.
\end{defn}

Note that a square matrix $A\in \R\left[\Dg_\aaa^\order\right]^{p\times p}$ is hyper-regular if, and only if, it is unimodular.

A straightforward adaptation of \cite[Section II.C]{Ant11} to $\Dg_\aaa^\order$-polynomial matrices reads:
 \begin{prop} \label{prop:Smith2}
(i) A matrix  $A\in \R\left[\Dg_\aaa^\order\right]^{p\times q}$, with $p<q$ is hyper-regular if, and only if,  it possesses a right-inverse, i.e. there exists $T$ in $ GL_{q}\(\R\left[\Dg_\aaa^\order\right]\)$ such that 
                    $AT = \left[ {I_p \tab 0_{p\times(q-p)}} \right]$.

(ii) A matrix $A\in \R\left[\Dg_\aaa^\order\right]^{p\times q}$, with $p\geq q$ is hyper-regular if, and only if, it possesses a left-inverse, i.e. there exists $S$ in $ GL_{p}\(\R\left[\Dg_\aaa^\order\right]\)$ such that 
                   $ SA =\left[ {\begin{array}{*{20}c}
                    I_q   \\
                    0 _{(p-q)\times q} \\
                    \end{array}} \right].$
 \end{prop}
{Note, again according to \cite{Ant11}, that the technique of row- or column-reduction described in this reference requires less computations for left- or right-inverses than the above mentioned Smith decomposition and might indeed be preferred.}

\section{Linear fractionally flat systems}\label{sec:Sortie_ plate_Matrice_poly}

\subsection{Linear fractional systems}\label{subsec:fracsys}
We consider a linear fractional system given by the following representation
    \begin{equation}\label{eq:poly_SSR}
         {\begin{array}{*{20}c}
           Ax = Bu 
        \end{array}} 
    \end{equation}
with state, or partial state, $x$ of dimension $n$,  input $u$ of dimension $m$, $A\in  \R\left[\Dg_\aaa^\order\right]^{n\times n} $ and $B\in  \R\left[\Dg_\aaa^\order\right]^{n\times m}$. $B$ is assumed to be of rank $m$, with $1\leq m\leq n$.
For System (\ref{eq:poly_SSR}), we consider (see e.g. \cite{Fli90a,Fli92a,Pol98,Tre92,Lev09,Ant14}):

\begin{itemize}
\item its \emph{behavior} $\ker \begin{bmatrix}A&-B\end{bmatrix}$, where the kernel is taken w.r.t. the signal space $\mathfrak{H}_\aaa $ defined in (\ref{eq:Ha}), i.e. the set $\left\{\begin{bmatrix}x\\u\end{bmatrix} \in \left( \mathfrak{H}_\aaa \right)^{n+m} \vert \begin{bmatrix}A&-B\end{bmatrix}\begin{bmatrix}x\\u\end{bmatrix} = 0 \right\}$;
\item and its \emph{system module} ${\mathfrak M}_{A,B}$, i.e. the quotient module
    \begin{equation}\label{eq:def_mod}
{\mathfrak M}_{A,B} = \R[\Dg_{\aaa}^{\order}]^{1\x (n+m)}/  \R[\Dg_{\aaa}^{\order}]^{1\x n}\begin{bmatrix}A&-B\end{bmatrix},
\end{equation}
 where 
$\R[\Dg_{\aaa}^{\order}]^{1\x (n+m)}$ is the set of row vectors 
with components in  $\R[\Dg_{\aaa}^{\order}]$ and where $\R\left[\Dg_\aaa^\order\right]^{1\x n}\begin{bmatrix}A,-B\end{bmatrix}$ is the module generated by the rows of the $n\x (n+m)$ matrix $\begin{bmatrix}A&-B\end{bmatrix}$.
\end{itemize}

According to \cite{Fli90a}, ${\mathfrak M}_{A,B}$ can be decomposed into the direct sum:
${\mathfrak M}_{A,B}={\mathcal T} \oplus {\mathcal F}$
where the uniquely defined module ${\mathcal T}$ is torsion and $\mathcal F$ is free. $\mathcal F$ is unique up to isomorphism.
Also, the system is said $F$-controllable
if and only if ${\mathcal T} =\{0\}$, or in other words, if and only if ${\mathfrak M}_{A,B}$ is free.

\subsection{Fractional flatness}\label{subsec:F_Flatness}

From now on, it is assumed that the matrix $F \triangleq \begin{bmatrix}A&-B\end{bmatrix} \in \R\left[\Dg_\aaa^\order\right]^{n \times (n+m)}$ has full (left) row rank. We also denote by $ {\mathfrak M}_{F}  \triangleq {\mathfrak M}_{A,B}$ the system module for simplicity's sake. 
System (\ref{eq:poly_SSR}) reads:
\begin{equation}\label{eq:syst_implicite}
F  \left[  \begin{array}{*{20}c}
                    x  \\
                    u
                    \end{array} \right]= 0.
\end{equation}
Our definition of \emph{fractional flatness} is based on the notion of \emph{defining matrices} in the spirit of \cite{Lev03,Ant14}.
\begin{defn}\label{def:defining_matrices}
The system (\ref{eq:syst_implicite}) is called fractionally flat if, and only if, there exist matrices $P\in \R\left[\Dg_\aaa^\order\right]^{m \times (n+m)}$ and $Q\in \R\left[\Dg_\aaa^\order\right]^{(n+m)\times m}$ such that
\begin{equation}\label{def:mat}
Q \left(\mathfrak{H}_\aaa\right)^m  = \ker F \quad \textrm{and} \quad 
PQ = I_m.
\end{equation}
\end{defn}

In other words, there exists a matrix $P$ with right-inverse $Q$ over the ring $\R\left[\Dg_\aaa^\order\right]$ such that, for all $(x,u)$ satis\-fying $ F\left[  \begin{array}{*{20}c} x \\ u \end{array} \right] =0 $, we have $ y=P\left[  \begin{array}{*{20}c} x \\ u \end{array} \right]$ and 
$ \left[  \begin{array}{*{20}c} x  \\ u \end{array} \right] = Q y$.
The variable $y$, taking its values in $\left(\mathfrak{H}_\aaa\right)^m$, is called \emph{fractionally flat output} and the matrices $P$ and $Q$ are called \emph{defining matrices}.

The main result of this section is the following\footnote{A comparable result, in the context of linear time-varying differential-delay systems, may be found in \cite{Ant14}.}:

\begin{thm}\label{th:controllability}
We have the following equivalences: 
\begin{description}
\item (i) system (\ref{eq:syst_implicite}) is fractionally flat;
\item (ii) the system module ${\mathfrak M}_{F}$ is free;
\item (iii) the matrix $F$ is hyper-regular over $\R\left[\Dg_\aaa^\order\right]$.
\end{description}
\end{thm}
\begin{proof}
(i) $\Longrightarrow$ (iii) Assuming system (\ref{eq:syst_implicite}) fractionally flat, according to Definition~\ref{def:defining_matrices}, there exist $P$ and $Q$ such that $Q \left(\mathfrak{H}_\aaa\right)^m  = \ker F$, or $FQ=0$, and $PQ = I_m$. According to the decomposition of morphisms (see e.g. \cite[A I, Th\'eor\`eme~3 p. 37]{Bou70}), since $Q$ is onto by definition, there exists an isomorphism of $\R\left[\Dg_\aaa^\order\right]$-modules $R: = \left( \mathfrak{H}_\aaa \right)^{n+m}/  \ker F \mapsto F  \left( \mathfrak{H}_\aaa \right)^{n+m}$ such that $F\left( \begin{matrix}R&Q\end{matrix}\right) = \left( \begin{matrix}I_{n}&0_{n\times m}\end{matrix} \right)$, which, together with the fact that  $R\in GL_{n}\(\R\left[\Dg_\aaa^\order\right]\)$ and that $Q$ has right-inverse, proves that $F$ is hyper-regular.

(iii) $\Longrightarrow$ (ii) If $F$ is hyper-regular, 
by Proposition~\ref{prop:Smith2}, 
there exists $T\in GL_{n+m}\(\R\left[\Dg_\aaa^\order\right]\)$ such that $FT= \left( \begin{matrix}I_{n}&0_{n\times m}\end{matrix} \right)$. Therefore the module ${\mathfrak M}_{F}$ is equi\-valent to the module ${\mathfrak M}_{FT} \simeq {\mathfrak M}_{ \left( \begin{matrix}I_{n}&0_{n\times m}\end{matrix} \right)}$ which is indeed free.

(ii) $\Longrightarrow$ (i) Assume that ${\mathfrak M}_{F}$ is free. According to the size of $F$, it admits $m$ independent variables $y_{1},\ldots, y_{m} \in \mathfrak{H}_\aaa$ as basis. Denote by $y=\left( y_{1},\ldots, y_{m}\right)^{T}$. If $\xi$ is an arbitrary element of ${\mathfrak M}_{F}$, we indeed have $y=T\xi$ and $\xi = S y$ for some suitable unimodular matrices $T$ and $S$. We readily get $TS=I_{m}$ and, since, by construction, $\xi= U\left[ \begin{array}{c}x\\u\end{array}\right]$ for some $\left[ \begin{array}{c}x\\u\end{array}\right]$ satisfying $F\left[ \begin{array}{c}x\\u\end{array}\right]=0$, i.e. $U^{-1}\xi \in \ker F$,  we get $U^{-1}S \left(\mathfrak{H}_\aaa\right)^{m}= \ker F$, thus proving that system (\ref{eq:syst_implicite}) is fractionally flat with $P= TU$ and $Q=U^{-1}S$. 
\end{proof}
\begin{rem}
Given (ii), a linear fractional system is flat if and only if it is F-controllable (see Section \ref{subsec:fracsys}).
\end{rem}

The following algorithm to compute a fractionally flat output is immediately deduced: 
\begin{alg}{Computation of fractionally flat output}\label{alg1}
~
\begin{description}
   \item[\textit{Input}:] The matrix $ F= \begin{bmatrix}A&-B\end{bmatrix} \in \R\left[\Dg_\aaa^\order\right]^{n\times (n+m)}$
   \item[\textit{Output}:]  ~~Defining matrices $P$ and $Q$, i.e. satisfying (\ref{def:mat}).   \item[\textit{Procedure}:]~
    	\begin{enumerate}
    	\item Use column-reduction to check if $F$ is hyper-regular. If not,  return  ``fail''.
    	\item Else, find $ W\in GL_{n+m}\(\R\left[\Dg_\aaa^\order\right]\)$, according to Proposition \ref{prop:Smith2} (i), such that
$ FW = \left[ I_n,\, 0_{n\times m}   \right] .$
	\item The defining matrices are given by:\\
$ Q\triangleq  W\left[  \begin{array}{*{20}c} 0_{n\times m}  \\ I_m \end{array} \right] ~ \textrm{and} ~  P\triangleq  \left[ 0_{m\times n},\, I_m \right]W^{-1}. $
  	\item Return $P$, $Q$ and a fractionally flat output given by $ y=P\left[  \begin{array}{*{20}c}x \\ u \end{array} \right]$.
	\end{enumerate}
\end{description}
\end{alg}
It is readily seen that this algorithm yields $FQ \left(\mathfrak{H}_\aaa\right)^m = 0$ and that $PQ= I_{m}$, which is precisely (\ref{def:mat}). Moreover, $x=  \left[ I_n,\, 0_{n\times m}   \right] Qy$ and $u= \left[ 0_{m\times n},\, I_{m}   \right] Qy$ identically satisfy (\ref{eq:syst_implicite}).

\subsection{Fractional 0-flatness}\label{subsec:0_F_Flatness}
In practical applications, it may be convenient to distinguish between state and input, 
and for linear controllable time-invariant systems, a set of flat outputs may be obtained via Brunovsk\'{y}'s canonical form (see e.g. \cite{Bru70,Kai80}) and do not depend on the input $u$. This property called \emph{fractional $0$-flatness}, reads: there exist $P_1\in \R\left[\Dg_\aaa^\order\right]^{m\times n} $ and $Q_1\in \R\left[\Dg_\aaa^\order\right]^{n\times m}$ such that $y=P_1x$, $x=Q_1y$, and $P_1Q_1=I_m$.  

\begin{defn}\label{def:0-flatness}
~
\begin{itemize}
\item A system is called fractionally $k$-flat, with $k\geq 1$, if and only if the maximal degree of the matrix $ P \begin{bmatrix} 0_{n,m}  \\ I_m \end{bmatrix}$ is equal to $k-1$. In this case, the output $y$ is called  \emph{fractionally $k$-flat output}.

\item It is said fractionally $0$-flat if $ P \begin{bmatrix}0_{n,m}  \\ I_m \end{bmatrix} = 0_{m}$. The associated output $y$ is called  \emph{fractionally $0$-flat output}.
\end{itemize}
\end{defn}
Fractional $0$-flatness is thus equivalent to the existence of $P$ and $Q$ as in \defref{def:defining_matrices} such that  $ P=\begin{bmatrix} P_1& 0_{m,n} \end{bmatrix}$ with $ P_1\in  \R\left[\Dg_\aaa^\order\right]^{m\times n} $ and $P_1Q_1=I_m$ where $ Q_1\triangleq \begin{bmatrix} I_n&0_{n,m} \end{bmatrix} Q$.

\begin{lem}[Elimination]\label{lem:elimination}
If $B$ is hyper-regular, there exists a unimodular matrix $M\in  \R\left[\Dg_\aaa^\order\right]^{n\times n}$ such that $ MB = \begin{bmatrix} I_{m}  \\0_{(n-m)\times m} \end{bmatrix}$. Moreover, there exist matrices $\tilde{F}\in  \R\left[\Dg_\aaa^\order\right]^{(n-m)\times n} $ and $R\in \R\left[\Dg_\aaa^\order\right]^{m\times n}$ such that System~(\ref{eq:poly_SSR})  is equivalent to $Rx=u$, $\tilde{F}x=0$.
\end{lem}
\begin{proof} 
Setting $ MA \triangleq  \begin{bmatrix} R  \\ \tilde{F} \end{bmatrix}$ with $R$ of size $m\times n$ and $\tilde{F}$ of size $(n-m)\times n$,
we have: $\begin{bmatrix}
                    R  \\
                    \tilde{F}
                    \end{bmatrix}x = MAx = MBu = \begin{bmatrix}
                    I_{m}  \\
                    0_{(n-m)\times m}
                    \end{bmatrix}  u$. 
\end{proof}

\begin{thm}\label{th:0-flat} 
If $B$ is hyper-regular, the following statements are equivalent:
\begin{description}
\item (i) System (\ref{eq:poly_SSR}) is fractionally $0$-flat;
\item (ii) The system module \\
\centerline{${\mathfrak M}_{\tilde{F}} \triangleq \R\left[\Dg_\aaa^\order\right]^{1\times n}/
\R\left[\Dg_\aaa^\order\right]^{1\times (n-m)} \tilde{F}$} 
is free, with $\tilde{F}$ defined in Lemma~\ref{lem:elimination};
\item (iii) the matrix $\tilde{F}$ is hyper-regular over $\R\left[\Dg_\aaa^\order\right]$.
\end{description}
\end{thm}
\begin{proof} Immediate consequence of Theorem~\ref{th:controllability}. See also \cite[Theorem~2]{Ant14} in a different context. 
\end{proof}

The previous algorithm is easily adapted to compute fractionally $0$-flat outputs:
\begin{alg}{Computation of fractionally $0$-flat output}\label{alg2}
~
\begin{description}
   \item[\textit{Input}:] The matrices $A$ and $B$ of System (\ref{eq:poly_SSR}) with $B$ hyper-regular.
   \item[\textit{Output}:]  \hspace{0.2cm} Defining matrices $P$ and $Q$, i.e. satisfying (\ref{def:mat}), with
   $P=\begin{bmatrix} P_1& 0_{m,m} \end{bmatrix}$, $P_1\in  \R\left[\Dg_\aaa^\order\right]^{m\times n} $,  $Q_1\triangleq \begin{bmatrix} I_n&0_{n,m} \end{bmatrix} Q$ and $P_1Q_1=I_m$.\hspace{0.3cm}
   \item[\textit{Procedure}:] \hspace{0.2cm}
    	\begin{enumerate}
    	\item Check, using row-reduction, if $B$ is hyper-regular. If not, return ``fail''.
    	\item Else, find $M\in GL_n\(\R\left[\Dg_\aaa^\order\right]\)$ such that 
$ MB = \begin{bmatrix}  I_m  \\ 0_{(n-m)\times m} \end{bmatrix}.$
	\item With $M$, obtain $R\in \R\left[\Dg_\aaa^\order\right]^{m\times n}$ and $\tilde{F}\in \R\left[\Dg_\aaa^\order\right]^{(n-m)\times n}$, according to Lemma~\ref{lem:elimination}, by:
$ MA= \begin{bmatrix} R \\ \tilde{F}  \end{bmatrix}. $
	\item Apply Algorithm 1 to $\tilde{F}$, which returns $P_1$ and $Q_1$. 
	\item Set $ P = \(P_1, \, 0\)$ and $ Q=  \begin{bmatrix} Q_1  \\  RQ_1  \end{bmatrix}$. 
	\end{enumerate}
\end{description}
\end{alg}

To summarize, Algorithm~\ref{alg2} gives a fractionally $0$-flat output $y$ by $y=P_{1}x$. Moreover, $x=Q_{1}y$ and $u=RQ_{1}y$ identically satisfy $Ax=Bu$.

\section{A thermal bidimensional application}\label{sec:thermal_app}
\subsection{Heat equation}
    \begin{figure}[ht]
        \begin{center}
         \psfrag{w}[][r]{\scalebox{0.9}{\quad $\varphi(y,t)$}}
         \psfrag{L2}[][l]{\scalebox{0.9}{$y$}}
         \psfrag{L1}[][l]{\scalebox{0.9}{$x$}}
         \psfrag{wy0}[][l]{\scalebox{0.9}{ }}
          \psfrag{wyL2}[][l]{\scalebox{0.9}{ }}
           \psfrag{wxL1}[][l]{\scalebox{0.9}{ }}
 \psfrag{T}[][l]{\scalebox{0.9}{$T(x_0,y_0,t)$}}
            \includegraphics[scale=0.45]{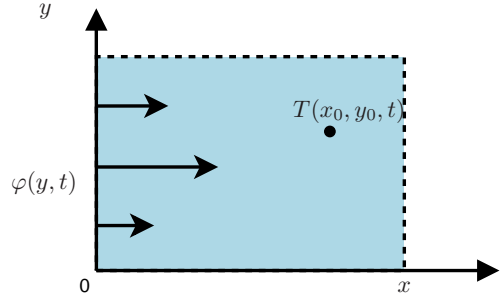}
            \caption{Heated metallic sheet}
            \label{fig:plaque_ch}
        \end{center}
    \end{figure}

We consider a 2D metallic sheet, assumed isolated and without heat losses (see \figref{fig:plaque_ch}), 
where the temperature $T(x,y,t)$ is controlled by the heat density flux $\varphi(y,t)$ across the $y$-axis for $y\geq 0$. The medium may be seen as a homogeneous metallic semi-infinite plane of diffusivity $\alpha$ and  conductivity $\lambda$. 
Our aim is to plan a temperature trajectory at a given point $(x_0,y_0)$ and to obtain the control that generates it. 

As it is well-known, the temperature diffusion satisfies the following scalar heat equation:
\begin{equation}\label{eq:heat_equation}
    \(\frac{\partial^2 }{\partial x^2}+ \frac{\partial^2}{\partial y^2}- \frac{1}{\alpha}\frac{\partial }{\partial t}\)T\(x,y,t\)=0, 
\end{equation}
in the open quarter-space $(x,y) \in ]0,\, \infty[\times ]0, \,\infty[$ and $t \in ]0,\, \infty[$, with boundary condition
\begin{equation}\label{eq:limit_conditions}
   -\lambda \frac{\partial T(x,y,t)}{\partial x}\big|_{x=0}=\varphi(y,t), \quad \forall y>0, \quad  \forall t>0
\end{equation}
$\varphi$ being the control variable, with the limit conditions
\begin{eqnarray}
\lim_{x\to \infty} T(x,y,t)=0, & \quad  \forall y>0, \quad \forall t>0, \label{eq:bound_cond}\\ 
\lim_{y\to \infty} T(x,y,t)=0, & \quad  \forall x>0, \quad \forall t>0, \label{eq:bound_cond2}
\end{eqnarray}
and Cauchy condition
\begin{equation}\label{eq:Cauchy}
    {T\left( {x,y,0} \right)}=0 \quad \forall x>0 , \quad \forall y>0.
\end{equation}

We follow a separation of variables approach developed, e.g.  in \cite{Cai05}
and generalizing the one followed by \cite{Lar98,Lar00,Sed01}. 
We recall that, for a function $f$ of several variables $\xi$ and $t\in \R$, its (partial) Laplace transform is given by
$\hat{f}(\xi,s) = \int_{0}^{+\infty} f(\xi,t) e^{-ts}dt$
for all $\xi$ and $s\in \Cc$ where this integral is finite.
Applying this transform to (\ref{eq:heat_equation}), we get:
\begin{equation}\label{eq:heat_equation_laplace}
     \frac{s}{\alpha}\hat{T}\(x,y,s\)=\frac{\partial^2\hat{T}\(x,y,s\) }{\partial x^2}+ \frac{\partial^2 \hat{T}\(x,y,s\)}{\partial y^2}.
\end{equation}

It may be seen that the formal series:
\begin{equation}\label{eq:T_P_HE_SV}
     \hat T(x,y,s)=\sum_{i=0}^{+\infty}L_{x,i}(s) L_{y,i}(s) e^{-\sqrt{s}\left(\frac{x}{i+1}+y\sqrt{\frac{1}{\alpha}-\frac{1}{(i+1)^2}}\right)}
\end{equation}
where $L_{x,i}(s)$ and $L_{y,i}(s)$ are arbitrary functions of the complex variable $s$,
satisfies (\ref{eq:heat_equation})
--(\ref{eq:Cauchy}) and that the heat flux, defined by (\ref{eq:limit_conditions}),  is given by 
 \begin{equation}\label{eq:fluxys}
   \hat \varphi(y,s)=\sum_{i=0}^{+\infty}\lambda\frac{\sqrt {s}}{i+1} L_{x,i}(s) L_{y,i}(s) e^{-\sqrt{s}y\sqrt{\frac{1}{\alpha}-\frac{1}{(i+1)^2}}}.
   \end{equation}
Setting 
\begin{equation}\label{eq:fluxys_basis} 
\hat\varphi_i (s) \triangleq \frac{\lambda\sqrt {s}}{i+1} L_{x,i}(s) L_{y,i}(s) , \quad \forall i,
\end{equation}
the unique solution of (\ref{eq:heat_equation_laplace}) for every given square integrable flux density $\hat\varphi$ reads:
\begin{equation}\label{eq:heat-sol}
\begin{aligned}
 \hat T(x,y,s)&=\sum_{i=0}^{+\infty} \frac{i+1}{\lambda\sqrt {s}}e^{-\sqrt{s}\left(\frac{x}{i+1}+y\sqrt{\frac{1}{\alpha}-\frac{1}{(i+1)^2}}\right)}\hat\varphi_i (s)\\
 \hat \varphi(y,s)&=\sum_{i=0}^{+\infty} e^{-\sqrt{s}\left(y\sqrt{\frac{1}{\alpha}-\frac{1}{(i+1)^2}}\right)}\hat\varphi_i (s).
 \end{aligned}
\end{equation}

 \begin{rem}
The sequence $\{\hat\varphi_i\}$ corresponds to a decomposition of the heat flow in spatial frequencies.
\end{rem}

For each $i\geq 0$, let us define the \emph{thermal impedance}:

\begin{equation}\label{eq:TF_2D}
    {H}_i(x,y,s)\triangleq \frac{(i+1) e^{-\sqrt{s}\left(\frac{x}{i+1}+y\sqrt{\frac{1}{\alpha}-\frac{1}{(i+1)^2}}\right)}}{\lambda\sqrt{s}},
\end{equation}
so that, denoting by $\hat{T}_{i}(x,y,s)= {H}_i(x,y,s)\hat\varphi_i$, we get
\begin{equation}\label{eq:TF_T}
 \hat{T}(x,y,s)=\sum_{i=0}^{+\infty}\hat{T}_{i}(x,y,s).
\end{equation}

\subsection{Approximate fractional heat transfer}
Using the Pad\'{e} approximant of $e^{-x}$ at the order $\mathbf{K}$ (see e.g. \cite{Bak96}), the expression (\ref{eq:TF_2D}) evaluated at the point $(x_{0},y_{0})$, 
with $\order = \frac{1}{2}$, reads:
\begin{equation}
\begin{aligned}
{H}_i (&x_{0},y_{0},s)= \frac{\frac{(i+1)}{\lambda} }{\sqrt{s}}e^{-\sqrt{s}\left(\frac{x_{0}}{i+1} + y_{0}\sqrt{\frac{1}{\alpha}-\frac{1}{(i+1)^2}}\right)}
\\    &\approx 
\frac{\sum_{k=0}^{\mathbf{K}} \frac{(i+1)}{\lambda} a_{i,k} s^{k\order}}{ \sum_{k=0}^{\mathbf{K}} |a_{i,k}|s^{(k+1)\order}} 
\triangleq  {H}_{i,\mathbf{K}}(x_{0},y_{0},s), \label{eq:TF_2D_NE2} 
\end{aligned} 
\end{equation}
with $a_{i,k}=\frac{(-1)^k(2\mathbf{K}-k)!\mathbf{K}!}{2\mathbf{K}!k! (\mathbf{K}-k)!} \left(\frac{x_{0}}{i+1}+y_{0}\sqrt{\frac{1}{\alpha}-\frac{1}{(i+1)^2}}\right)^k $.

For a point $(x_0,y_0)$ not too far from the origin, for every $i\geq 0$, ${H}_{i,\mathbf{K}}$
 fastly converges  to ${H}_i$ as $\mathbf{K}$ tends to infinity. 
Also, given a finite number $\mathbf{I}$, the truncated temperature 
\begin{equation}\label{eq:TF_T_IK}
\hat{T}_{\mathbf{I},\mathbf{K}}(x_0,y_0,s) \triangleq  \sum_{i=0}^{\mathbf{I}}{H}_{i,\mathbf{K}}(x_0,y_0,s) \hat\varphi_{i}(s)
\end{equation}
may be seen to  fastly converge to $\hat{T}(x_0,y_0,s)$ as $\mathbf{I}$ tends to infinity. 
Note that the finite sequence $\{ \varphi_{i}, i= 0,\ldots, \mathbf{I} \}$ now plays the role of the control vector. 

In state space form, the transfer (\ref{eq:TF_T_IK}) 
can be represented by:
\begin{equation}\label{sys:therm_app} 
AX = BU, \quad T_{\mathbf{I},\mathbf{K}}(x_0,y_0,t) = CX, 
\end{equation}
where: 
  \begin{equation}\label{eq:XU}
  \begin{array}{c}
       X \triangleq \begin{bmatrix} X_0\\ \vdots \\ X_{\mathbf{I}} \end{bmatrix},~
       U \triangleq \begin{bmatrix} \varphi_0\\ \vdots \\ \varphi_{\mathbf{I}} \end{bmatrix} ,\\
       A \triangleq \diag{A_i} ,  B\triangleq\diag{B_i },  C \triangleq  \begin{bmatrix} C_0, \ldots, C_{\mathbf{I}}\end{bmatrix}
       \end{array}
\end{equation}   
with, for $i=0,\ldots, \mathbf{I}$,     
 \begin{equation}\label{eq:A}
 X_i  \triangleq \begin{bmatrix} X_{i,\mathbf{K}}\\ \vdots\\ X_{i,0}  \end{bmatrix},
\end{equation}
 \begin{equation*}
  A_i  \triangleq \left[ {\begin{array}{*{20}c}
                \Dg_\aaa^\frac{1}{2} + | a'_{i,\mathbf{K}-1}|  & | a'_{i,\mathbf{K}-2}| &       \ldots      &   |a'_{i,0} |  &    0    \\
                   -1  &  \Dg_\aaa^\frac{1}{2}        & 0 & \ldots &   0\\
                    0 &   -1  & \ddots  &   \ddots& \vdots\\
                    \vdots & \ddots & \ddots & \ddots & 0 \\
                    0 &  \ldots &  0 &   -1 & \Dg_\aaa^\frac{1}{2}
             \end{array}} \right],
\end{equation*}
\begin{equation}\label{eq:BC}
       B_i  \triangleq 
       \left[  \begin{array}{c} 1  \\  0 _{\mathbf{K}\times 1} \end{array} \right]
       ,
        C_i \triangleq  
        \frac{(i+1)}{\lambda} \left[ \begin{array}{ccc} a'_{i,\mathbf{K}}& ,\cdots,&a'_{i,0} \end{array} \right],
\end{equation}
 where we have denoted $a'_{i,k}= a_{i,k}/|a_{i,\mathbf{K}}|$ for all $k=0, \ldots, {\mathbf{K}}$ and all $i=0,\ldots, \mathbf{I}$.
        
\subsection{Flat output computation}
Let us now apply Algorithm~2. Since every hyper-regular $B_i$ is already in its Smith form, we introduce the matrices $ R_i\in \R\left[\Dg_\aaa^\frac{1}{2}\right]^{1\times (\mathbf{K}+1)}$ and $\tilde{F_i}\in \R\left[\Dg_\aaa^\frac{1}{2}\right]^{\mathbf{K}\times (\mathbf{K}+1)}$ such that:
$ A_i =\left[ \begin{matrix}   R_i  \\
                 \tilde{F_i}\end{matrix} \right] $
     with            
     $$ R_i = \left[ \Dg_\aaa^\frac{1}{2}+| a'_{i,\mathbf{K}-1}| , \, | a'_{i,\mathbf{K}-2}| ,     \,  \ldots      ,  \,  |a'_{i,0} |  ,\,    0       \right],$$
$$ \tilde{F_i} =\left[ {\begin{array}{*{20}c}
                   -1  &  \Dg_\aaa^\frac{1}{2}        &  &    0\\
                     &   \ddots  &  \ddots &  \\
                    0 &   &   -1 & \Dg_\aaa^\frac{1}{2}
             \end{array}} \right], \quad i=0,\ldots, \mathbf{I}.
$$

To get the required implicit form, we consider the unimodular matrix $ M=\begin{bmatrix}\diag{M_{R,i},i=0,\ldots, \mathbf{I}}\\ \diag{M_{\tilde{F},i},i=0,\ldots, \mathbf{I}}\end{bmatrix}$
with $M_{R,i}=[1,0_{1, \mathbf{K}}]$ and $M_{\tilde{F},i}=[0_{\mathbf{K},1}, I_{\mathbf{K}}]$ for all $i=0,\ldots, \mathbf{I}$. and we can verify that $ MA=\begin{bmatrix}\diag{R_i,i=0,\ldots, \mathbf{I}}\\\diag{\tilde{F}_i,i=0,\ldots, \mathbf{I}}\end{bmatrix}\triangleq \begin{bmatrix}R\\\tilde{F} \end{bmatrix}$, $\tilde{F}$ being hyper-regular, and thus System (\ref{sys:therm_app})-(\ref{eq:XU})-(\ref{eq:A})-(\ref{eq:BC}) is $0$-flat.

Applying Algorithm 1 to $\tilde{F}$, we compute the upper triangular matrix $W\in GL_{(\mathbf{I}+1)(\mathbf{K}+1)}\(\R\left[\Dg_\aaa^\order\right]\)
$ satisfying
$\tilde{F} W = \begin{bmatrix} I_{\mathbf{K}(\mathbf{I}+1)}& 0_{\mathbf{K}(\mathbf{I}+1)\times (\mathbf{I}+1)} \end{bmatrix}, $
$$ W = \diag{W_i,i=0,\ldots, \mathbf{I}}, \text{with }$$ 
$$W_i\triangleq\left[ {\begin{array}{*{20}c}
                   -1  &  -\Dg_\aaa^\frac{1}{2}  &     -\Dg_\aaa^{1}   & \ldots &    -\Dg_\aaa^{\frac{K}{2}}\\
                     &   \ddots  & \ddots& \ddots & \vdots \\
                     &   &  \ddots &\ddots &  -\Dg_\aaa^{1}  \\
                     &   &   &  \ddots &  -\Dg_\aaa^\frac{1}{2}\\
                    0 &  &  &  & -1
             \end{array}} \right],\,  i=0,\ldots,\mathbf{I}.$$
and its inverse, also upper triangular,
\[ W ^{-1}= \diag{W_i^{-1},i=0,\ldots, \mathbf{I}}, \text{with}\]
\[ W_i^{-1}\triangleq
\left[ {\begin{array}{*{20}c}
                   -1  &  \Dg_\aaa^\frac{1}{2} &    &  0\\
                     &   \ddots  & \ddots& \\
                     &     &  \ddots &  \Dg_\aaa^\frac{1}{2}  \\
                    0 &   &  & -1 
             \end{array}} \right], \quad i=0,\ldots,\mathbf{I}.\]

Therefore, by introducing $S_i=\begin{bmatrix} 0 _{\mathbf{K}\times 1}\\ 1 \end{bmatrix} $, $Q_{1,i}=\begin{bmatrix}  -\Dg_\aaa^{\frac{\mathbf{K}}{2}}\\  \vdots\\ -\Dg_\aaa^\frac{1}{2}\\ - 1 \end{bmatrix}$ and $P_{1,i} = \begin{bmatrix} 0_{1\times \mathbf{K}}& 1\end{bmatrix} $, for $i=0,\ldots, \mathbf{I}$:
$$ \begin{aligned} Q_1 &= W \diag{S_i} =  
\diag{Q_{1,i},i=0,\ldots, \mathbf{I}} ,\\
P_1  &=  \diag{P_{1,i}}W^{-1} = \diag{-P_{1,i},i=0,\ldots, \mathbf{I}},
\end{aligned}$$
we indeed have $P_1Q_1 = I_{\mathbf{I}+1}$.
Finally, the defining matrices $P$ and $Q$ read:
$$
 P =  \diag{[-P_{1,i},0],i=0,\ldots, \mathbf{I}},\quad 
Q  =  \begin{bmatrix} Q_1  \\ RQ_1 \end{bmatrix} $$
which proves that System (\ref{sys:therm_app}) is fractionally $0$-flat and that a fractional flat output $Y$ is given by
\begin{equation*}
 Y=P_1 X=-\begin{bmatrix}X_{0,0}\\ \vdots\\ X_{\mathbf{I},0}\end{bmatrix},
 \end{equation*}
 \begin{equation*}
U = RQ_1 Y = \begin{bmatrix} - {\sum\limits_{k = 0}^{\mathbf{K}} {\left| {a'_{i,k} } \right| \Dg_\aaa^{\frac{k+1}{2}} Y_0} } \\ \vdots\\ - {\sum\limits_{k = 0}^{\mathbf{K}} {\left| {a'_{i,k} } \right| \Dg_\aaa^{\frac{k+1}{2}} Y_I} }\ \end{bmatrix},
\end{equation*}
\begin{align}\label{eq:flatout-thermex}
 T_{\mathbf{I},\mathbf{K}}(x_0,y_0,t)&= CX= CQ_1 Y\\
\nonumber  &=-\sum_{i=0}^{\mathbf{I}} \frac{i+1}{\lambda}\sum_{k=0}^{\mathbf{K}}{a'_{i,k} \Dg_\aaa^{\frac{k}{2}}Y_{i}}.
\end{align}
In particular $Y_i=-X_{i,0}$, $i=0,\ldots, \mathbf{I}$.

\begin{rem}
The original infinite dimensional system (\ref{eq:heat_equation})-(\ref{eq:Cauchy}), once expressed in the frequency domain, may be seen to be flat with the infinite sequence 
$Y(s)\triangleq (L_{x,1}(s)L_{y,1}(s), \ldots, L_{x,i}(s)L_{y,i}(s), \ldots)$
as flat output, since $\hat{T}(x_0,y_0, s)$ and $\hat{\varphi}(y,s)$ may be expressed in terms of $Y(s)$  (see (\ref{eq:T_P_HE_SV})-(\ref{eq:fluxys_basis})). The interpretation of these expressions in the time domain is far from being obvious due to the sum of convolutions appearing in the inverse Laplace transform of the products  $L_{x,i}(s)L_{y,i}(s)e^{-\sqrt{s}}$.
Accordingly, the $0$-flat output (\ref{eq:flatout-thermex}) has no clear physical interpretation
in terms of the heat equation original variables. 
\end{rem}

\begin{rem}
For the infinite dimensional system
(\ref{eq:heat_equation})-(\ref{eq:Cauchy}), 
the temperature and the heat flux must be chosen in the class of Gevrey functions in order to gua\-rantee the convergence of the series (\ref{eq:fluxys_basis})-(\ref{eq:heat-sol}) (see e.g. \cite{Lar98,Lar00,Rud03,Sed01}).
 Dealing with fractional operators provides an important simplification in the design, though convergence aspects may be confirmed numerically. It can be verified numerically that the series $\hat{T}_{\mathbf{I},\mathbf{K}}$ is fastly convergent. Therefore, the orders of truncation ${\mathbf{I}}$ and ${\mathbf{K}}$ may be chosen small.
\end{rem}

\subsection{Trajectory planning}
Let us define a rest-to-rest temperature trajectory at the point $x_0=0.045$m and $y_0=0.02$m (see \figref{fig:plaque_ch}) with a temperature rise of $30^\circ C$ over the ambient temperature in a total duration $t_f$, with the conditions:
\begin{align}\label{eq:condtions}
 T(x_0,y_0,0)&\triangleq T_0 =0, \\ 
\nonumber T^{(l)}(x_0,y_0,0)&=0, \quad l=1,\ldots,\mathbf{L},\\
\nonumber T(x_0,y_0,t_f)&\triangleq T_f =30, \\ 
\nonumber T^{(l)}(x_0,y_0,t_f)&=0,\quad l=1,\ldots,\mathbf{L},
 \end{align} 
with $\mathbf{L}$  to be determined later.

The desired temperature $T_{\mathbf{I,K}}(x_0,y_0,t)$ and the control $U$ will be deduced from (\ref{eq:flatout-thermex}), \emph{without need to integrate the system equation (\ref{sys:therm_app})},
by translating the conditions (\ref{eq:condtions}) into conditions on the fractionally flat outputs $Y_i$, $i=0,\ldots,\mathbf{I}$, assuming that their trajectory is given by a polynomial of $t$:
\begin{equation}\label{eq:flatinterpol}
Y_i(t) \triangleq \sum_{j=0}^{\mathbf{r}} \eta_{i,j}\left( \frac{t}{t_f} \right)^{j} 
\end{equation}
with the integer $\mathbf{r}$ and the real coefficients $\eta_{i,j}$, $i=0,\ldots,\mathbf{I}$, $j=0,\ldots, \mathbf{r}$, to be determined. 

We compute the initial and final values $Y_i(0)$ and $Y_i(t_f)$ using (\ref{eq:flatout-thermex}) 
and knowing that, for $ t\geq 0$:
$$\Dg_\aaa^{\frac{k}{2}}t^j = \frac{\Gamma(j+1)}{\Gamma(j+1-\frac{k}{2})} t^{j-\frac{k}{2}},\, k=0, \ldots, \mathbf{K}, \, j=0,\ldots, \mathbf{r}.$$
We also assume that we can identify $T(x_0,y_0,t)$ with $T_{\mathbf{I},\mathbf{K}}(x_0,y_0,t) $.
Thus:
\begin{equation*}
\begin{array}{l} 
T(t) \triangleq T(x_0, y_0,t) \approx T_{\mathbf{I},\mathbf{K}}(x_0,y_0,t)\\
= -\sum_{i=0}^{\mathbf{I}} \frac{i+1}{\lambda}
 \sum_{j=0}^{\mathbf{r}}
\frac{\eta_{i,j}}{t_{f}^{j}}
 \sum_{k=0}^{\mathbf{K}}
a'_{i,k} \frac{\Gamma(j+1)}{\Gamma(j+1-\frac{k}{2})} t^{j-\frac{k}{2}},\\
 T^{(l)}(t)\approx T_{\mathbf{I},\mathbf{K}}^{(l)}(x_0,y_0,t)\\
 =-\sum_{i=0}^{\mathbf{I}} \frac{i+1}{\lambda}
\sum_{j=0}^{\mathbf{r}}
 \frac{\eta_{i,j}}{t_{f}^{j}}\sum_{k=0}^{\mathbf{K}}
a'_{i,k} \frac{\Gamma(j+1)}{\Gamma(j+1-\frac{k}{2}-l)} t^{j-\frac{k}{2}-l}.
\end{array}
\end{equation*}

In order to satisfy the initial conditions ($t=0$),
the coefficients $\eta_{i,j}$, corresponding to the negative exponents ($j-\frac{k}{2}-l<0$), must vanish. Therefore $\eta_{i,j} = 0$ for all $j < \frac{\mathbf{K}}{2}+\mathbf{L}$ and $T$ reads:
{\small$$T(t)= -\sum_{i=0}^{\mathbf{I}}\frac{i+1}{\lambda}
 \sum_{j=\lceil \frac{\mathbf{K}}{2}+\mathbf{L} \rceil}^{\mathbf{r}}
 \frac{\eta_{i,j}}{t_{f}^{j}} \sum_{k=0}^{\mathbf{K}}
a'_{i,k} \frac{\Gamma(j+1)}{\Gamma(j+1-\frac{k}{2})} t^{j-\frac{k}{2}},
$$}
where $\lceil. \rceil$ is the ceiling operator. 
Since, there are $\mathbf{L}+1$ final conditions left for $T^{(l)}(t_f)$, $l=0,\ldots,\mathbf{L}$, $\mathbf{r}$ must satisfy  $\left(\mathbf{r}-\lceil  \frac{\mathbf{K}}{2}+\mathbf{L}\rceil +1\right)(\mathbf{I}+1)\geq 2(\mathbf{L}+1)$, thus
 \begin{equation}\label{eq:rbound}
  \mathbf{r} \geq 2\frac{\mathbf{L}+1}{\mathbf{I}+1}+ \left\lceil \frac{\mathbf{K}}{2}+\mathbf{L} \right\rceil -1.
 \end{equation}

 \begin{figure*}[ht]
  \begin{center}
        \psfrag{t}[][l]{\scalebox{0.9}{time (s)}}
        \psfrag{z}[][l]{\scalebox{0.9}{}}
        \psfrag{flat}[][l]{\scalebox{0.9}{Flat outputs $Y_i$}}
        \psfrag{data1}[][l]{\scalebox{0.75}{\quad\quad $Y_0$}}
        \psfrag{data2}[][l]{\scalebox{0.75}{\quad\quad$Y_1$}}
            \includegraphics[scale=0.59]{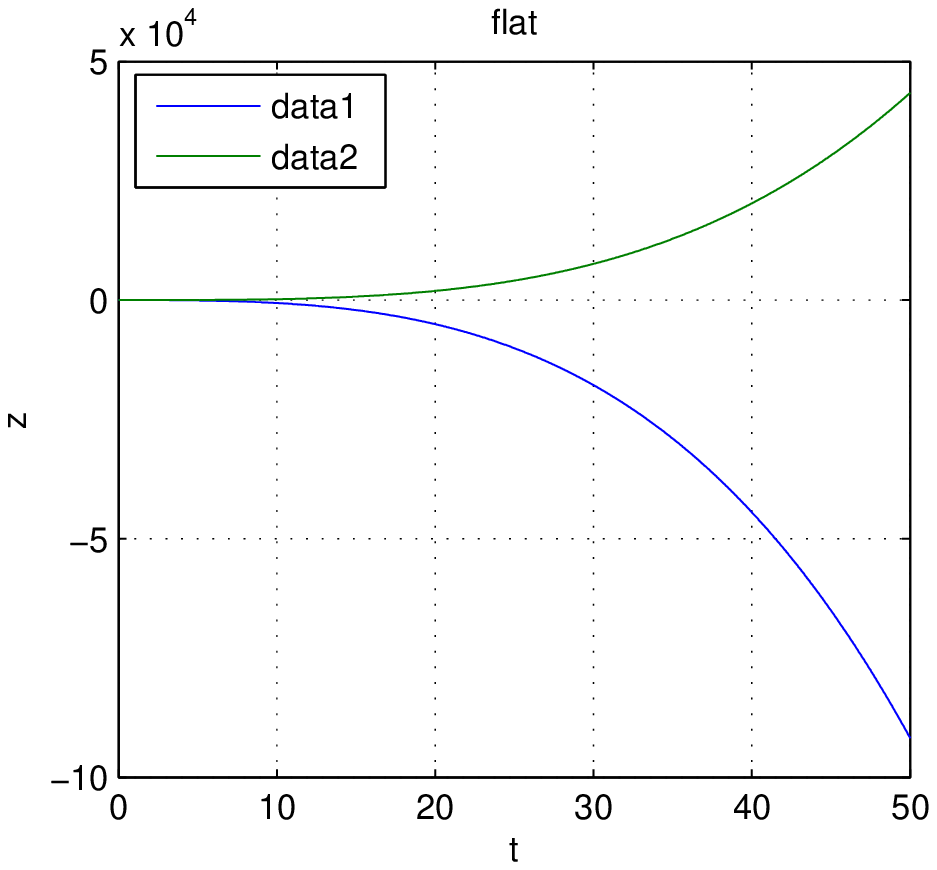}            
\hspace{-0.5cm}
        \psfrag{t}[][l]{\scalebox{0.9}{time (s)}}
        \psfrag{u}[][l]{\scalebox{0.9}{u}}
        \psfrag{flux}[][l]{\scalebox{0.9}{controls $\varphi_i$}}
        \psfrag{data2}[][l]{\scalebox{0.75}{\quad\quad $\varphi_0$}}
        \psfrag{data3}[][l]{\scalebox{0.75}{\quad\quad$\varphi_1$}}
        \psfrag{data1}[][l]{\scalebox{0.75}{\quad\quad$\varphi$}}
            \includegraphics[scale=0.59]{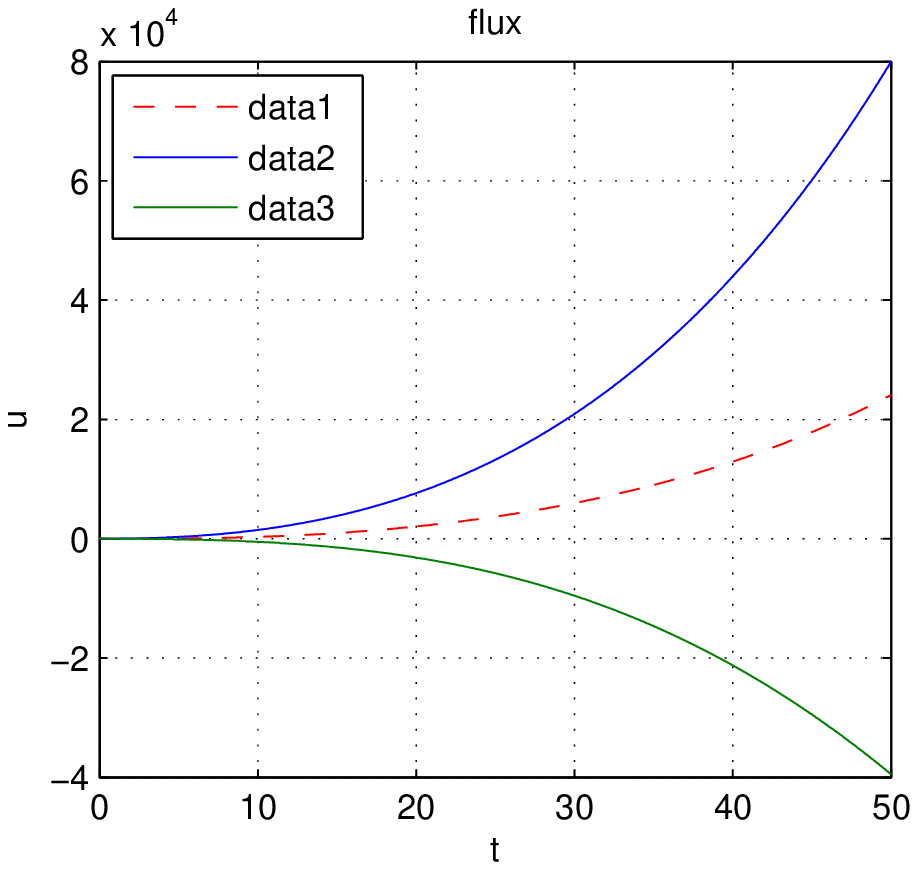}
\hspace{-0.5cm}      
               \psfrag{data1}[][l]{\scalebox{0.75 }{ $T$}}
        \psfrag{data2}[][l]{\scalebox{0.75}{\quad\quad $T_{exact}$}}
        \psfrag{t}[][l]{\scalebox{0.9}{time (s)}}
        \psfrag{T}[][l]{\scalebox{0.9}{$T(x_0,y_0,t)$ (°C)}}
        \psfrag{temp}[][l]{\scalebox{0.9}{Temperature}}
            \includegraphics[scale=0.59]{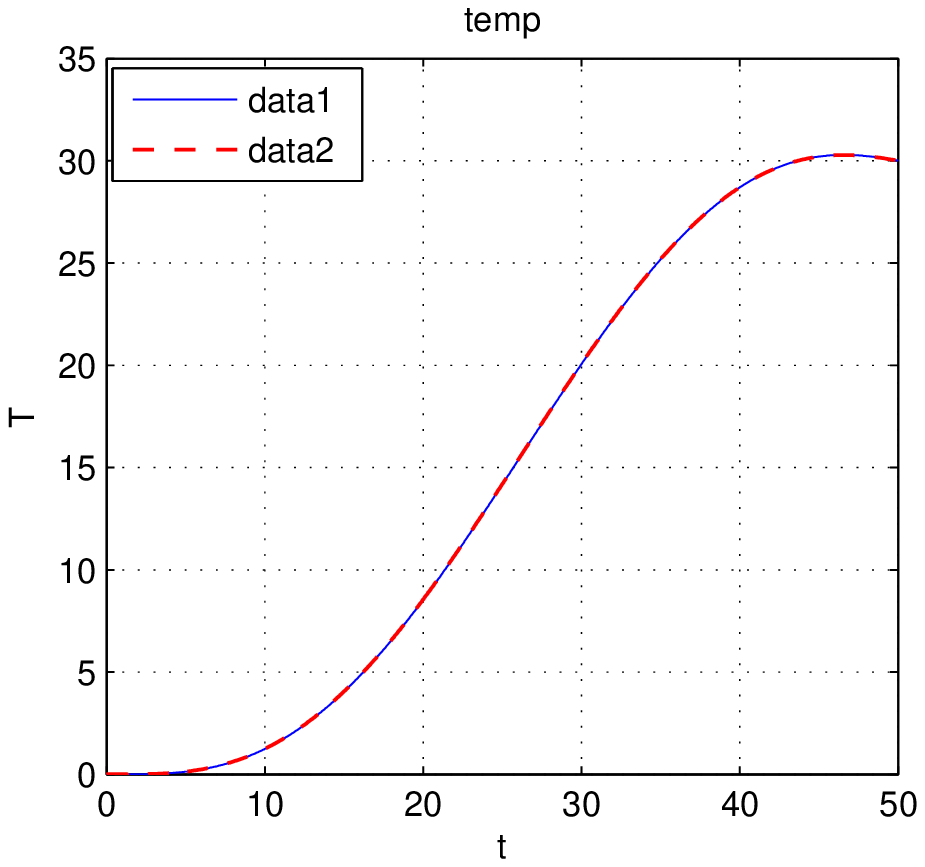}
            \caption{Left: Flat outputs $Y_i, \quad i=0,1$. Center: Decomposition of the heat flow (control) $\varphi_i(t), \quad i=0,1$ with the heat density flux (dashed red). Right: Temperature trajectory.}
            \label{fig:sim}
    \end{center}
\end{figure*}

We therefore have to solve the linear system in the coefficients $\eta_{i,j}$, $j=\lceil \frac{\mathbf{K}}{2}+\mathbf{L} \rceil, \ldots \mathbf{r}$, $i=0,\ldots,\mathbf{I}$:
 \begin{equation}\label{eq:T_l_tf}
T^{(l)}(t_f)=-\sum_{i=0}^{\mathbf{I}} \frac{i+1}{\lambda}
\sum_{j=\lceil \frac{\mathbf{K}}{2}+\mathbf{L} \rceil}^{\mathbf{r}}
\frac{\eta_{i,j}}{t_{f}^{j}}\sum_{k=0}^{\mathbf{K}}
a'_{i,k} \frac{\Gamma(j+1)}{\Gamma(j+1-\frac{k}{2}-l)} t_f^{j-\frac{k}{2}-l}
=0,
 \end{equation} 
 \begin{equation}\label{eq:T_0_tf}
 T(t_f)=-\sum_{i=0}^{\mathbf{I}}\frac{i+1}{\lambda}
 \sum_{j=\lceil \frac{\mathbf{K}}{2}+\mathbf{L} \rceil}^{\mathbf{r}}
 \frac{\eta_{i,j}}{t_{f}^{j}} \sum_{k=0}^{\mathbf{K}}
a'_{i,k} \frac{\Gamma(j+1)}{\Gamma(j+1-\frac{k}{2})} t_f^{j-\frac{k}{2}}
=T_f,
 \end{equation}
for all $l=1,\ldots,\mathbf{L}$, and plug its solution in (\ref{eq:flatinterpol}).

 \subsection{Simulations}
In the simulations presented in Figure \ref{fig:sim}, $\alpha= 8.83\times\, 10^{-5}m^2.s^{-1}$, $\lambda = 210 W.m^{-1}.K^{-1}$ and $t_f=50s$. We have chosen $\aaa =0$,  $\mathbf{K}=\mathbf{L} =2$,  $\mathbf{I}=1$ and $\mathbf{r}=5$. 
The coefficients $a'_{i,k} = \frac{a_{i,k}}{| a_{i,3} |}$ with $a_{i,k}=\frac{(-1)^k(2\mathbf{K}-k)!\mathbf{K}!}{2\mathbf{K}!k! (\mathbf{K}-k)!} \left(\frac{x_{0}}{i+1}+y_{0}\sqrt{\frac{1}{\alpha}-\frac{1}{(i+1)^2}}\right)^k $ in System (\ref{eq:T_l_tf})-(\ref{eq:T_0_tf}) are thus:
\[
\begin{array}{llll}
	a'_{0,1}= 2.543 &~~ a'_{0,2}= -2.762  &~~ a'_{0,3}= 1 \\
	a'_{1,1}= 2.596 &~~ a'_{1,2}= -2.791 &~~ a'_{1,3}= 1 \\
\end{array}
\]

The obtained coefficients $\eta_{i,j}$, $i=0,1$, $j=3,4,5$,  are:
\[
\begin{array}{llll}
	\eta_{0,3} =-7.30. 10^{4} & \eta_{0,4} =-1.13.10^{3}& \eta_{0,5} =-7.52.10^{3} \\
	\eta_{1,3} = 1.84.10^{4}& \eta_{1,4} =3.23.10^{4} & \eta_{1,5} =-7.33.10^{3}  \\
\end{array}
\]

The fractional flat output is depicted in Figure \ref{fig:sim} (left) and the corresponding heat flow (control) decomposition ($\varphi_i(t)$)  in Figure \ref{fig:sim} (center).
It can be noticed that the duration $t_f=50$s to reach the point of coordinates $x_0=0.045$m and $y_0=0.02$m is small compared to the time response of the system. However, the total flux density $\varphi$ of about $2.10^{4}W.m^{-2}$ (Figure \ref{fig:sim}, center, dashed line) needed to increase the temperature of $30^\circ C$, remains reasonable.
In Figure~\ref{fig:sim} (right), the exact solution of the heat equation, computed by applying the inverse Laplace transform to (\ref{eq:heat-sol}), is denoted by $T_{exact}$ and is plotted in dashed line. It is compared to the computed reference trajectory $T_{\mathbf{I},\mathbf{K}}$ (continuous line). Note that the error between $T_{\mathbf{I},\mathbf{K}}$ and $T_{exact}$ remains small, less than $0.02°C$, during the transient and converges to 0, which confirms the validity of our fractional approximation, even with such a rough truncation (recall that $\mathbf{K}=\mathbf{L} =2$ and $\mathbf{I}=1$). 
    
{\begin{rem}    
As it is well-known, the change of time scale $\tau= \alpha t$ results in normalizing the heat equation, i.e. $\frac{\partial^{2}T}{\partial x^{2}} + \frac{\partial^{2}T}{\partial y^{2}} - \frac{\partial T}{\partial \tau} = 0$.
Therefore, a small $\alpha$ will result in a slow temperature evolution compared to the spatial diffusion and conversely for large $\alpha$, with infinitely differentiable dependence. 
Moreover, if the diffusivity $\alpha$ is not precisely known, since steady states do not depend on it, as well as the flat output  $Y_i=-X_{i,0}$, $i=0,\ldots, \mathbf{I}$, as far as rest-to-rest trajectories are concerned, a small error on this coefficient might only mildly affect the state transient.
\end{rem}}

\section{Conclusions}\label{sec:conclu}
The notion of differential flatness has been extended to fractional linear  systems. A simple characterization of fractional flat outputs and fractional 0-flat outputs has been obtained in terms of polynomial matrices in the operator $\Dg_\aaa^{\order}$, leading to a simple algorithm to compute them. These results have been applied to a fractional approximation of order $\frac{1}{2}$ of the heat equation, correspon\-ding to a model of a heated 2-dimensional metallic sheet. The open-loop input flux density that gene\-rates the desired temperature profile at a given point of the metallic sheet has been deduced from the flat output trajectory, without integration of the system PDE.
Simulations have been presented, showing in particular that the obtained trajectory and the (exact) solution of the heat equation, computed by inverse Laplace transform, are very close, even with a rough truncation of the series expansion. It is remarkable that such a fractional model provides an important simplification in the trajectory design compared to the one based on the PDE model which requires using Gevrey functions to ensure the convergence of some series, though, with our approach, convergence aspects may be confirmed numerically.

For future works, feedback controllers for this model will be studied such as CRONE controllers, to improve the robustness of the trajectory tracking versus various perturbations.

\appendix
\section{Recalls on the properties of the fractional derivative}
\subsection{Fractional integrability}
Recall that $\nu\in [0,1[$, and $t\in[a,+\infty[$.
Let us set
\begin{equation}\label{eq:gfnctn}
 g_{t}^{\n}(\tau) = \frac{1}{\Gamma(\n + 1)} \left( t^{\n} - (t-\tau)^{\n}\right).
\end{equation}
The Stieltjes measure $dg_{t}^{\n}(\tau) \mathop  =  \limits^\Delta \frac{\partial g_{t}^{\n}}{\partial \tau} (\tau) d\tau =\frac{1}{\Gamma(\n)}$ $ \left( t-\tau \right)^{\n -1} d\tau$ is indeed absolutely continuous with res\-pect to the Lebesgue measure $d\tau$, and its density $\frac{\partial g_{t}^{\n}}{\partial \tau}$ belongs to $L^{1}(\aaa,t)$ for all $t \geq \aaa$ since
{$$
\begin{aligned}
\big\Vert \frac{\partial g_{t}^{\n}}{\partial \tau} \big\Vert_{L^{1}(\aaa,t)} &
= \int_{\aaa}^{t} \big| \frac{\partial g_{t}^{\n}}{\partial \tau}(\tau) \big| d\tau = g_{t}^{\n}(t) - g_{t}^{\n}(\aaa) 
\\
&
= \frac{ \vert t-\aaa \vert^{\n}}{\Gamma(\n + 1)} < +\infty.
\end{aligned}
$$}
With these notations, the fractional primitive $\Ig_{\aaa}^{\n}f(t)$ reads
\begin{equation}\label{eq:stieljes}
\Ig_{\aaa}^{\n}f(t) =  \int_{\aaa}^{t} f(\tau) dg_{t}^{\n}(\tau).
\end{equation}
Since
$\vert \Ig_{\aaa}^{\n}f(t) \vert \leq \big\Vert \frac{\partial g_{t}^{\n}}{\partial \tau} \big\Vert_{L^{1}(\aaa,t)} \cdot \Vert f \Vert_{L^{\infty}(\aaa,t)} < +\infty$,
we immediately deduce that the fractional primitive
is well-defined for all $f\in L^{\infty}(\aaa,t)$, and \emph{a fortiori} for all $f \in  \C^{\infty}([\aaa,+\infty[)$, all $\n \in [0,1[$ and all finite $t\geq \aaa$.
Moreover, the definition (\ref{eq:stieljes}) may be extended to any order $\n+k$ for an arbitrary integer $k$:
\begin{equation}\label{eq:stieljes_nk}
 \Ig_{\aaa}^{\n+k}f(t) =  \int_{\aaa}^{t} f(\tau) dg_{t}^{\n+k}(\tau).
\end{equation}

\subsection{Fractional differentiability }
(see \cite[Chap.2, Sec. 2.2.5, p.57]{Pod99a})

\begin{prop}\label{prop:DnInu_InuDn}
The operator $\Dg_{\aaa}^{\order}$ maps $\C^{\infty}([\aaa,+\infty[)$ to $\C^{\infty}(]\aaa,+\infty[)$  and 
\begin{equation}\label{eq:gamma-diff}
\Dg_{\aaa}^{\order}f(t) = \Ig_{\aaa}^{\n}f^{(n)}(t) + \sum_{j=0}^{n-1} \frac{(-1)^{j}(t-\aaa)^{\n -n + j}}{\Gamma(\n -n + j + 1)} \delta_{\aaa}^{(j)}f.
\end{equation}
with $k=n$ and $\gamma = n -\n$. 
Moreover, for all $f\in \mathfrak{H}_\aaa$ and $\order=n-\n$, we have
\begin{equation}\label{eq:gamma-diff1}
\Dg_{\aaa}^{\order}f(t) = \Ig_{\aaa}^{\n}\Dg^n f(t).
\end{equation}
\end{prop}

\subsection{Commutativity  }
(see \cite[Chap.2, Sec. 2.2.6, p.59]{Pod99a})

\begin{prop}[Commutativity of $\Dg^\order_\aaa$]\label{prop:Dcom}
If $f\in \mathfrak{H}_\aaa$, we have
\begin{equation}\label{eq:Dmu_Dorder_commutativity2}
    \Dg^\kappa_\aaa \(\Dg^\order_\aaa f(t)\)=\Dg^\order_\aaa \(\Dg^\kappa_\aaa f(t)\)= \Dg_\aaa^{\order+\kappa} f(t).
\end{equation}
\end{prop}

\end{document}